\documentclass{amsart}
\pdfoutput=1
\usepackage[utf8]{inputenc}
\usepackage{mathptmx}
\usepackage{lmodern}
\usepackage[T1]{fontenc}
\usepackage[english]{babel}
\usepackage{mathtools}
\usepackage{etoolbox}
\usepackage{etex}
\usepackage{fixltx2e}
\usepackage{xparse}
\usepackage[autostyle = true, english = american, strict = true, autopunct=true]{csquotes}
\usepackage[final, babel]{microtype}
\usepackage{xspace}
\usepackage{amssymb}
\usepackage{amsthm}
\usepackage{multicol}
\usepackage{enumitem}
\usepackage{booktabs}
\usepackage{fnpct}
\usepackage{array}
\usepackage{hyperref}
\definecolor{persimmon}{rgb}{0.93, 0.35, 0.0}
\definecolor{cssgreen}{rgb}{0.0, 0.5, 0.0}
 \definecolor{bleudefrance}{rgb}{0.19, 0.55, 0.91}
 \hypersetup{plainpages=false, bookmarksnumbered=true, bookmarksopen=true, bookmarksdepth=2, colorlinks=true, breaklinks=true, linkcolor=persimmon, urlcolor=bleudefrance, citecolor=cssgreen, pdfdisplaydoctitle= true}
\usepackage[hyperref=true, natbib=true, backend=bibtex8, maxnames=20, maxitems=20, texencoding=utf-8, useprefix=true, block=space, sorting=nyt, abbreviate=true, firstinits=true, isbn=false, url=false]{biblatex}
\newtheorem{theorem}{Theorem}
\newtheorem{definition}{Definition}
\allowdisplaybreaks %
\newcounter{tabenum}
\renewcommand{\thetabenum}{\alph{tabenum})}
\newcommand{\nextnum}{\refstepcounter{tabenum}\thetabenum}

\newcommand{\sizeof}[1]{|#1|}
\newcommand{\tridro}{\vartriangleright}
\newcommand{\trigau}{\vartriangleleft}
\newcommand{\state}[1]{\textbf{#1}\xspace}
\newcommand{\cc}[1]{\ensuremath{\mathbf{#1}}}
\newcommand{\onstack}[1]{\texttt{#1}\xspace}
\newcommand{\pop}{\onstack{pop}}
\newcommand{\peek}{\onstack{peek}}
\newcommand{\push}{\onstack{push}}
\newcommand{\bstack}{\boxdot}

\makeatletter
\newcommand{\proofstep}[1]{%
 \addvspace{\medskipamount}%
 #1\@addpunct{:}\enspace\ignorespaces
}
\makeatother

\ExplSyntaxOn
\NewDocumentCommand{\fa}{ s t+ m o o}{
	\IfNoValueTF{#4}{\bool_set_false:N \l_show_heads_bool}{\bool_set_true:N \l_show_heads_bool \tl_set:Nn \l_heads_tl {#4}}
	\IfNoValueTF{#5}{\bool_set_false:N \l_show_stacks_bool}{\bool_set_true:N \l_show_stacks_bool \tl_set:Nn \l_stacks_tl {#5}}
	\IfBooleanTF{#2}{\fa_long:n { #3 }}
	{\IfBooleanTF{#1}{\ensuremath{\mathrm{\cc{\fa_short:n { #3 }}}}}{\ensuremath{\mathrm{\fa_short:n { #3 }}}}}
	}
\bool_new:N \l_show_heads_bool
\tl_new:N \l_heads_tl
\bool_new:N \l_show_stacks_bool
\tl_new:N \l_stacks_tl
\cs_new_protected:Npn \fa_short:n #1{\fa_short_aux:NNNN #1} %
\cs_new_protected:Npn \fa_short_aux:NNNN #1 #2 #3 #4{
	#1#3FA(\bool_if:NT \l_show_heads_bool {\l_heads_tl}, \bool_if:NT \l_show_stacks_bool {\l_stacks_tl})
}
\cs_new_protected:Npn \fa_long:n #1{\fa_long_aux:NNNN #1} %
\cs_new_protected:Npn \fa_long_aux:NNNN #1 #2 #3 #4{%
	\str_case:nn { #1 }{{1}{\(1\)-way~}{2}{\(2\)-way~}{\_}{}}
	\str_case:nn { #3 }{{N}{non-deterministic}{D}{deterministic}{\_}{}}finite~automaton\xspace
	\bool_if:NTF \l_show_heads_bool{with~\(\l_heads_tl\) -heads~}{\str_case:nn { #2 }{{S}{single-head~}{M}{multi-head~}{\_}{}}}
	\bool_if:NTF \l_show_stacks_bool{
and~\(\l_stacks_tl\)~pushdown~stacks}{\str_case:nn { #4 }{{+}{with~pushdown~stack}{-}{without~pushdown~stack}{\_}{}}}
}
\ExplSyntaxOff

\newcommand{\titred}{An in-between “implicit” and “explicit” complexity: Automata}
\title{\titred}
\author[C. Aubert]{Clément Aubert}
\address{INRIA, Université Paris-Est, LACL (EA 4219), UPEC, F-94010 Créteil, France}
\email{\href{mailto:clement.aubert@lacl.fr}{clement.aubert@lacl.fr}}
\urladdr{\url{https://lacl.fr/~caubert/}}
\thanks{This work was partly supported by the ANR-10-BLAN-0213 Logoi, the ANR-14-CE25-0005 ELICA and the ANR-11-INSE-0007 REVER}
\date{\today}

\hypersetup{
	pdfinfo={
			Author={Clément Aubert},
			title={\titred},
			Keywords={Implicit Complexity, Automata Theory, Logarithmic Space, Polynomial Time, Stack, Survey, History},
			Creator={PDFLaTeX},
			Producer={PDFLaTeX},
			Lang={en}}
}

\begin{document}
\begin{abstract}
Implicit Computational Complexity (ICC) makes two aspects implicit, by manipulating programming languages rather than models of computation, and by internalizing the bounds rather than using external measure.
We survey how automata theory contributed to complexity with a machine-dependant with implicit bounds model.
\end{abstract}

\maketitle

\section*{Introduction}
This survey justifies a fine-grained view on the definition of \emph{what} Implicit Computational Complexity (ICC) wants to keep implicit:

\begin{center}
\resizebox{\textwidth}{!}{
\begin{tabular}{
l%
@{\quad} %
l%
@{\quad} %
l%
}
& \textbf{Machine-dependant} & \textbf{Machine-independant}\\[.5em]
\textbf{Explicit bounds} &
\begin{tabular}{@{}l@{}}Turing machine,\\ Random access machine, \\ Counter machine, \ldots \\\end{tabular} &
\begin{tabular}{@{}l@{}}Bounded recursion on notation \cite{Cobham1965},\\ Bounded arithmetic \cite{Buss1986},\\ Bounded linear logic \cite{Girard1992}, \ldots\\[.5em] \end{tabular}\\
\textbf{Implicit bounds} &
\begin{tabular}{@{}l@{}}Automaton,\\Auxiliary pushdown machine \cite{Cook1971},\\Boolean circuit, \ldots \\\end{tabular} &
\begin{tabular}{@{}l@{}}Descriptive complexity \cite{Fagin1973},\\Recursion on notation \cite{Bellantoni1992},\\ Tiered recurrence \cite{Leivant1993}, \ldots \\\end{tabular}
\end{tabular}
}
\end{center}

It is common to refer to the top-left as the \emph{classical}, or \emph{explicit}, approach to complexity.
The three others techniques \emph{keep something implicit}.
The bottom-right way is considered as the \emph{true} implicit complexity, the one that produces all the motivating perspectives: quasi-interpretations, non-size-increasing computation, soft lambda-calculus and linear logics, to name a few.
But there might be hesitations when defining ICC, for neither the \enquote{machine-independent} neither the \enquote{without explicit bounds} slogans are precise enough.

The machine-dependant with implicit bounds variant will be our subject.
It characterizes complexity classes by \emph{bounding the \emph{quality} of the resources of a machine rather than its \emph{quantity}}.
To our knowledge, automata are the most canonical model of this enterprise, which dates back to the 70's:
\textquote[{\cite[p.~88]{Ibarra1971}}]{we have attempted to characterize several tape and time complexity classes of Turing machines in terms of devices whose definitions involve only ways in which their infinite memory may be manipulated and no restrictions are imposed on the amount of memory that they use.}

\section*{Automata, characterizations and separation results}

\begin{definition}%
\label{ndfa_def}
For \(k \geqslant 1\), \(l \geqslant 0\), a \emph{\fa+{2MN+}[k][l] (\(\fa{2MN+}[k][l]\))}
is a tuple \(M = \{ \state{S}, \state{i}, \state{F}, A, B, \tridro, \trigau, \bstack, \sigma\}\) where:
\begin{itemize}
\item \(\state{S}\) is the finite set of \emph{states}, with \(\state{i} \in \state{S}\) the \emph{initial state} and \(\state{F} \subseteq \state{S}\) the set of \emph{accepting states};
\item \(A\) is the \emph{input alphabet}, \(B\) is the \emph{stack alphabet};
\item \(\tridro\) and \(\trigau\) are the \emph{left} and \emph{right endmarkers}, \(\tridro, \trigau \notin A\);
\item \(\bstack\) is the \emph{bottom symbol of the stack}, \(\bstack \notin B\);
\end{itemize}
From now on, we let \(A_{\bowtie}\) (resp. \(B_{\bstack}\)) be \(A \cup \{ \tridro, \trigau\}\) (resp. \(B \cup \{ \bstack\}\)).
\begin{itemize}
\item \(\sigma \subseteq (\state{S} \times (A_{\bowtie})^k \times (B_{\bstack})^l) \times (\state{S} \times \{-1, 0, +1\}^k \times \{\pop, \peek, \push(b)\}^l) \) is the \emph{transition relation}, where \(-1\) means to move the head one cell to the left, \(0\) means to keep the head on the current cell and \(+1\) means to move it one cell to the right.
Regarding the pushdown stacks, \(\pop\) means \enquote{erase the top symbol}, \(\peek\) \enquote{do nothing}, and, for all \(b \in B\), \(\push(b)\) is \enquote{write \(b\) on top of the stack}.
\end{itemize}
Given an input \(n \in A^*\), the automata is initiated in state \state{i}, with \(\tridro n \trigau\) written on its only tape, all its heads at \(\tridro\) and all its stacks containing \(\bstack\).
It makes (non-deterministic) transitions according to \(\sigma\) and halt accepting as soon as it reaches a state belonging to \(\state{F}\).
We impose that the heads cannot move beyond the endmarkers and that the bottom stack symbol \(\bstack\) cannot be erased (\enquote{popped}).

Finally, we denote with \(\fa*{2MN+}[k][l]\) the class of languages recognized by the \(\fa{2MN+}[k][l]\) automaton.
Moreover, we let \(\fa*{2MN+}[*][l] = \cup_{k \geqslant 1} \fa*{2MN+}[k][l]\).
\end{definition}

\begin{theorem}
The following table gives the correspondence between automata and languages (or predicates):

\begin{center}
\begin{tabular}{>{\nextnum}l@{\hspace{\labelsep}}*{2}{l@{\quad\quad}}l}
\multicolumn{1}{c}{} & \textbf{Automata} & \textbf{Language}\\[.5em]
\label{tableau-comput} & \(\fa*{2MN+}[1][2]\) & Computable\\
\label{tableau-p} & \(\fa*{2MN+}[*][1]\) & Polynomial time\\
\label{tableau-l} & \(\fa*{2MN+}[*][0]\) & Logarithmic space\\
\label{tableau-context} & \(\fa*{2MN+}[1][1]\) & Context-free\\
\label{tableau-regular} & \(\fa*{2MN+}[1][0]\) & Regular
\end{tabular}
\end{center}
\end{theorem}

\begin{proof}
We just sketch them, and provide references for the most difficult one.
We always suppose that \(n\) is the input and \(\sizeof{n}\) its size.

\proofstep{\ref{tableau-comput} \(\subseteq\)} Given a computable language, by the Church–Turing thesis, there exists a Turing machine \(T\) that decides it.
Using some classical theorem, we can always assume that \(T\) has a single reading head and a single read-write tape.
A \fa{2MN+}[1][2] can simulate \(T\) by simulating the movement of the read-only head with its head, and the content of the read-write tape with its two pushdown stacks.
The first (resp. second) pushdown stack store the content on the left (resp. right) of the read-write head.

\proofstep{\ref{tableau-comput} \(\supseteq\)} As \fa{2MN+}[1][2] are restrictions of Turing Machines, their simulation is obvious.

\proofstep{\ref{tableau-p} \(\subseteq\)}
This part amounts to designing an equivalent Turing machine whose movements of heads follow a regular pattern.
That permits to seamlessly simulate the content of the read-write tape with a pushdown stack.
A complete proof~\cite[pp.~9--11]{Cook1971} as well as a precise algorithm~\cite[pp.~238--240]{Wagner1986} can be found in the literature.

\proofstep{\ref{tableau-p} \(\supseteq\)}
This way gave birth to memoization.
In a nutshell, simulating a \(\fa{2MN+}[*][1]\) with a polynomial-time Turing machine cannot amount to simulate step-by-step the automaton.
The reason is that for any automaton, one can design an automaton that recognizes the same language but runs exponentially slower \cite[p.~197]{Aho1968}.
The technique invented by Alfred V.~Aho et. al~\cite{Aho1968} and made popular by Stephen A.~Cook consists in building a \enquote{memoization table} that allows the Turing machine to create shortcuts in the simulation of the automaton, decreasing drastically its computation time.
It is a \enquote{\emph{clever evaluation strategy}, applicable whenever the results of certain computations are needed more than once}~\cite[p.~348]{Amtoft1992}.
A nice explanation in the case of single head automata is in a recent and short article by R. Glück~\cite{Gluck2013}.

\proofstep{\ref{tableau-l} \(\subseteq\)}
Let \(T\) be a Turing machine deciding the predicate with \(m \times \log(\sizeof{n})\) space.
Remark that an integer represented by a binary string of length \(p\) is no greater than \(2^p\), which takes \(2^p\) bits in unary.
Then, a binary string of length \(\log (\sizeof{n})\) cannot be greater than \(\sizeof{n}\).
So the content of the read-write tape and the position of its head is encoded as distances between \(\tridro\) and the read-only heads of the automata.
That needs to compute simple operations (modulo, multiplication and remaining of a division) thanks to an extra head, but \(m+3\) heads \cite[pp.~191--192]{Sudborough1977} are sufficient to perform this simulation.
Simplest and the fastest simulations exists~\cite[pp.~223--225]{Wagner1986}, but the overhead in terms of heads is more important.

\proofstep{\ref{tableau-l} \(\supseteq\)} The addresses of the heads on the input \(n\) takes \(\log (\sizeof{n})\) to we written, hence a \(\log\)-space Turing machine can easily simulate a \(\fa{2MN+}[*][0]\).

\proofstep{\ref{tableau-context} and \ref{tableau-regular}}
The proofs of those two fundamentals results are beyond the scope of this small survey.
Their original proofs are respectively due to \citeauthor{Chomsky1962} \cite[Theorem~1, p.~188]{Chomsky1962} and \citeauthor{Kleene1956} \cite[Theorem~6]{Kleene1956}.

\end{proof}

In the \ref{tableau-l} case, taking a \emph{deterministic} automata \(\fa{2MD-}[*][0]\) yields a characterization of \emph{deterministic} \(\log\)-space predicate \cc{L} (whereas we characterized here the non-deterministic case \cc{NL}).
In the \ref{tableau-p} case, both deterministic \emph{and} non-deterministic automata characterize \emph{deterministic} polynomial time.

Automata can also be restricted to be \(1\)-way (simply remove the \(-1\) instruction from the transition relation).
In the \ref{tableau-regular} case of regular language (which are equal to linear space), we get in fact \(\fa*{1MD-}[1][0] = \fa*{2MD-}[1][0] = \fa*{1MN-}[1][0] = \fa*{2MD-}[1][0]\).

\begin{theorem}[Other results of interest]
\begin{multicols}{2}
\begin{align*}
\fa*{1MD-}[k][0] & \subsetneq \fa*{2MD-}[k][0]	\tag{\cite[p.~95]{Holzer2008}}\\
\fa*{1MN-}[k][0] &\subsetneq \fa*{2MN-}[k][0]	\tag{\cite[p.~95]{Holzer2008}}\\
\fa*{1MD-}[k][0] & \subsetneq \fa*{1MN-}[2][0]	\tag{\cite[p.~339]{Yao1978}}\\
\fa*{1MD+}[k][1] & \subsetneq \fa*{1MN+}[k][1]	\tag{\cite{Chrobak1986}}\\
\fa*{1MN-}[k][0] & \subseteq \fa*{2MD+}[k][0]	\tag{\cite[p.~190]{Sudborough1977}}\\
\fa*{2MN+}[k][1] & \subseteq \fa*{2MD+}[4 \times k][1] \tag{\cite[p.~214]{Sudborough1977a}}\\
\fa*{2MN-}[k][0] & \subseteq \fa*{2MN+}[2\times k][1] \tag{\cite[p.~189]{Sudborough1977}}\\
\fa*{1MD+}[k][1] & \subsetneq \fa*{1MN+}[k][1]	\tag{\cite{Chrobak1986}}\\
\cc{NL} &= \cc{L} \tag{\cite[pp.~75--76]{Sudborough1973}} \\
& \text{~iff~} \\
\fa*{1MN-}[2][0] & \subseteq \fa*{2MD-}[k][0]
\end{align*}
\end{multicols}
\end{theorem}

\begin{theorem}
For \(l < 2\), \(k+1\) heads are strictly stronger than \(k\) heads.
\end{theorem}

This results holds in all situations,  whenever automata are \(1\)- or \(2\)-way, deterministic or not, with or without pushdown stacks%
\footnote{To be fair, the author could not find a proof that this is the case for \(1\)-way non-deterministic finite automata with one stack.}%
.
Even if some notes on the techniques and historic of those results exist \cites[Chap.~4, Sect.~5.3]{Vitanyi1990}[pp.~67--68]{Monien1980}, no extensive survey was found, so the reader has to go for the original research papers \cites[p.~338]{Yao1978}[p.~106]{Monien1976}[p.~383]{Monien1977}[p.~179]{Chrobak1986}[p.~35]{Ibarra1973}.

\section*{Conclusion}
Automata theory is a major subject of computer science, very active and providing all kind of extensions and restrictions to automata.
It contributed to \enquote{implicit} complexity in the sense that no external bound or measure is needed to tame their computational power: by specifying their number of ways, heads and stacks, one knows \emph{in advance} the computational power of the device.

Those \enquote{implicit characterizations} provided decisive hints and tools to design two (truly implicit) bounded programming languages \cite{Aubert2014b, Bagnol2014}.
Numerous other inspiring results remains to be explored and should benefit to ICC.

We did not mentioned finite state transducer, which are functional finite automata, but they also provide nice characterizations of functional complexity and beautiful results (even recently \cite{Filiot2013}, where the main theorem fits in the title).
The status of the input is also of interest: automata can bee feeded with \(1\)-way inputs, trees, etc.
That makes the expressive power to be really parametric in the input \emph{and} the machine

\renewcommand*{\bibfont}{\footnotesize}
\printbibliography

\section*{Acknowledgements}
The author wish to thank, in no particular order,
M. Bagnol,
G. Bonfante,
A. Durand,
P. Jacobé de Naurois,
U. dal Lago,
P. Pistone,
P.-A. Reynier,
U. Schöpp,
T. Seiller
and
the \href{http://tex.stackexchange.com/}{\TeX{}-\LaTeX{} Stack Exchange} community.
\end{document}